%% file: main.tex
\documentclass[runningheads]{llncs}

\input{preamble}

\newif\ifnotes\notesfalse
\ifnotes
 \definecolor{mygrey}{gray}{0.50}
 \newcommand{\notenamem}[2]{{\textcolor{mygrey}{\footnotesize{\bf (#1:} {#2}{\bf ) }}}}
 \newcommand{\noteswarning}{{\begin{center} {\Large WARNING: NOTES ON}\end{center}}}

 \else

 \newcommand{\notenamem}[2]{{}}
 \newcommand{\noteswarning}{{}}

 \fi

 \newcommand{\znote}[1]{{\notenamem{Zeta}{#1}}}
 
 \newcommand{\inote}[1]{{\notenamem{Iosif}{#1}}}
 \newcommand{\stnote}[1]{{\notenamem{Samarth}{#1}}}

\begin{document}

\title{Wiser: Increasing Throughput in Payment Channel Networks\\with Transaction Aggregation}

\titlerunning{Wiser: Transaction Aggregation in PCNs}

\author{Samarth Tiwari\inst{1}\thanks{Supported by ERC Starting Grant QIP--805241.} \and Michelle Yeo\inst{2} \and
Zeta Avarikioti\inst{3} \and Iosif Salem\inst{4} \and
Krzysztof Pietrzak\inst{2}\thanks{Supported by the Vienna Cybersecurity and Privacy Research Center (ViSP), funded by the Vienna business agency (Wirtschaftsagentur),  2020-2023.} \and
Stefan Schmid\inst{5}\thanks{Supported partially by the Austrian Science Fund (FWF) project I 4800-N (ADVISE), 2020-2023 and the Vienna Cybersecurity and Privacy Research Center (ViSP), funded by the Vienna business agency (Wirtschaftsagentur),  2020-2023.}
}

\authorrunning{S. Tiwari et al.}

\institute{Centrum Wiskunde \& Informatica, Amsterdam, The Netherlands\\
\email{samarth.tiwari@cwi.nl} \and
Institute of Science and Technology Austria \email{\{michelle.yeo, krzysztof.pietrzak\}@ist.ac.at} \and
TU Wien, Austria\\
\email{georgia.avarikioti@tuwien.ac.at} \and
TU Berlin, Germany\\
\email{iosif.salem@inet.tu-berlin.de} \and
TU Berlin, Germany \& University of Vienna, Austria\\
\email{stefan.schmid@tu-berlin.de}
}

\maketitle
\begin{abstract}

Payment channel networks (PCNs) are one of the most prominent solutions to the limited transaction throughput of blockchains. Nevertheless, PCNs suffer themselves from a throughput limitation due to the capital constraints of their channels. A similar dependence on high capital is also found in inter-bank payment settlements, where the so-called netting technique is used to mitigate liquidity demands.

In this work, we alleviate this limitation by introducing the notion of transaction aggregation: instead of executing transactions sequentially through a PCN, we enable senders to aggregate multiple transactions and execute them simultaneously to benefit from several amounts that may ``cancel out''. 
Two direct advantages of our proposal is the decrease in intermediary fees paid by senders as well as the obfuscation of the transaction data from the intermediaries.

We formulate the transaction aggregation as a computational problem, a generalization of the Bank Clearing Problem.  We present a generic framework for the transaction aggregation execution, and thereafter we propose \sys as an implementation of this framework in a specific hub-based setting.
To overcome the NP-hardness of the transaction aggregation problem, in \sys we propose a fixed-parameter linear algorithm for a special case of transaction aggregation as well as the Bank Clearing Problem.
\sys can also be seen as a modern variant of the \textit{Hawala} money transfer system, as well as a decentralized implementation of the overseas remittance service of Wise.

\end{abstract}

\noteswarning

\section{Introduction}
\subsection{Motivation}
Payment channel networks (PCNs)~\cite{poon2015lightning,decker2015fast,decker2018eltoo,miller2017sprites} are a promising technology to significantly increase the transaction throughput of blockchain-based cryptocurrencies such as Bitcoin~\cite{nakamoto2008bitcoin}. Two parties can open a \emph{payment channel} by locking coins in a joint account on-chain. Thereafter, the parties can transact off-chain via updating the distribution of the coins in the channel to depict the new transaction. Only if the parties want to close the channel or allocate extra budget, do they need to go on-chain. Therefore, payment channels can facilitate arbitrary many transactions off-chain with a constant  number of on-chain transactions, significantly increasing  the transaction throughput.
These payment channels collectively constitute a \emph{payment channel network} (PCN) which allows two parties to transact off-chain even if they do not share a payment channel by using other parties as intermediaries in routing their transactions~\cite{poon2015lightning,decker2015fast}.

In order to provide high and reliable transaction throughput, PCNs require many well connected high-capacity payment channels. The reason is that channels may quickly deplete if several transactions are forwarded along a channel or path in the same direction, thus preventing any further use of the channel in that direction.
However, locking large amount of capital for a long period of time comes with a cost. To mitigate this cost, the users of PCNs ask for a service fee for relaying transactions of others through their channels.
As a result, a business opportunity arises: wealthy individuals can act as \emph{(transaction) hubs} by establishing well-connected high-capacity nodes, offering routing of transactions as a service. In this way, capital reserves can be used to generate passive income.

While this business model benefits the PCN's users, it also introduces several challenges:
(a) the liquidity of the hubs still poses a constraint to the PCN's overall liquidity and hence limits the transaction throughput,
(b) the transaction fees that are typically proportional to the transaction value may be high, and
(c) the hubs may learn a significant amount of transactions, thus compromising the privacy of users.

To address these challenges, we introduce the notion of \emph{transaction aggregation} (Figure~\ref{fig:transAggrIntro}). In high level, we enable users to aggregate multiple transactions and execute them simultaneously in the PCN instead of sequentially, i.e., one-by-one. This way multiple transactions may ``cancel out'' effectively reducing the transferred amounts and the corresponding transaction fees, thus increasing the PCN's liquidity and consequently the transaction throughput. 
Moreover, the transaction aggregation obfuscates the individual transactions routed through a channel, as only the aggregated amount is executed.

A closely related notion to transaction aggregation in finance is \emph{netting} that deals with the aggregation of various financial obligations, say, to mitigate risk. 
In particular, the application of netting in inter-bank  payment systems essentially tackles the first aforementioned challenge. The computation problem of optimal netting, i.e., optimal aggregation of payments across traditional banks, is known as the Bank Clearing Problem or BCP. 
The transaction aggregation problem  we introduce in this work is a more generic formulation of the BCP problem, hence demonstrating the potential impact of our work in the financial landscape even beyond blockchains.

Transaction aggregation can also be seen as a decentralized and privacy-enhancing variant of the payment system Wise (formerly TransferWise)~\cite{wise} built on payment channel networks. Wise  is a financial technology allowing convenient online transactions, often used to cheaply transfer money abroad. Another similar financial system is \emph{Hawala}~\cite{hawala}, a centuries-old system of fund transfer across long distances and overseas. A system that implements transaction aggregation can be interpreted as a modern and improved digital version of the \emph{Hawala} money transfer system, which is still in use in parts of the Middle East and South Asia.

\subsection{Contribution}
The goal of this work is to address the three challenges of PCNs in the hub-based business model, i.e., (a)
improve the throughput of payment channel networks while (b) minimizing fees in the process and also (c) accounting for privacy concerns. 
To do so, we propose a system which allows hubs to act cooperatively, combining their capital to serve all their clients. 
Any two users connected to any of the hubs can thereafter transact thought the hubs. 
The transactions collected over a fixed time interval are then aggregated into a single monetary flow through the network. This flow is consequently executed atomically through the PCN, meaning that either all the selected transactions are executed or none. 
We term the computational problem of finding a subset of transactions that can be aggregated into a feasible resultant flow as the \textit{transaction aggregation problem} (Figure \ref{fig:transAggrIntro}). 

Specifically, we present a generic framework which combines three building blocks: (a)~a protocol to enable optimal liquidity usage among hubs, (b)~an efficient algorithm for the transaction aggregation problem, and (c)~a protocol for atomic multi-channel state updates, i.e., the atomic execution of the selected transaction in the PCN.
We present a skeleton protocol, that is an abstraction of the composed protocols, to illustrate its benefits as well as technical challenges. We further determine the necessary properties for such a system and show that our abstraction satisfies them. 

We then present \sys, a concrete implementation of the generic framework. \sys uses channel factories~\cite{burchert2017scalable} and Thora~\cite{aumayr2022thora} as substantiations of the aforementioned building blocks ((a) and (c) respectively). As the final building block (b), we present a fixed-parameter linear\footnote{\FPL\ is a subclass of the fixed-parameter tractable (\textsc{FPT}) complexity class \cite{downey2012parameterized}.} algorithm to solve the transaction aggregation problem, which we first show to be NP-hard. For the execution of our algorithm, we leverage secure multi-party computation (MPC) to conceal channel balances and transaction data.


We further derive a connection between transaction aggregation and netting, and exploit this connection to enrich both notions. On one hand, our fixed-parameter linear algorithm developed for \sys is directly applicable to BCP. Our algorithm exhibits the best known parameterized complexity, improving dependence on number of transactions from pseudo-polynomial to linear. On the other hand, algorithms for the Bank Clearing Problem with inferior asymptotic complexity may still be faster in certain parameter regimes. \sys can be implemented with any of these optimization algorithms as a building block, so that work on BCP also bolsters the versatility of our protocol.

\begin{wrapfigure}[17]{R}{0.5\textwidth}
    \centering
    \includegraphics[width = 0.4\textwidth]{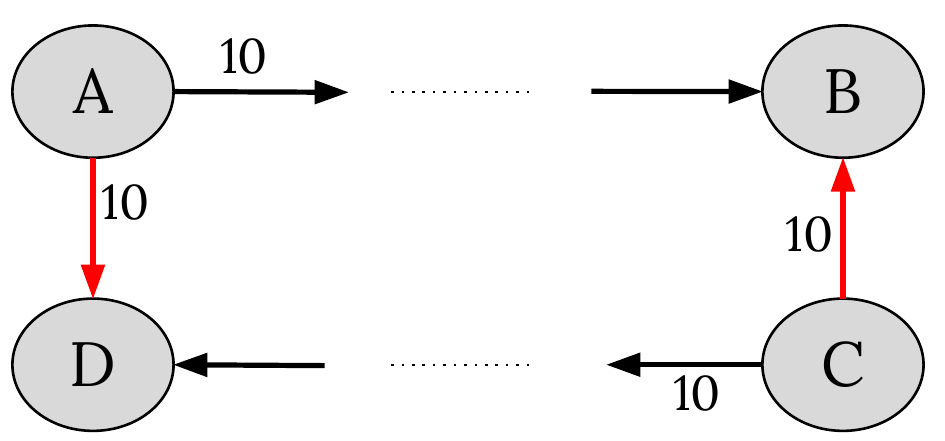}
    \caption{Transaction aggregation example. User A wants to make a payment of 10 coins to user B and user C wants to make a payment of 10 coins to user D. The A-to-B and C-to-D payment paths might render the transaction infeasible and incur routing fees linear to the path length. Instead, if A pays D and C pays B, each 10 coins, the transactions success depends on the capacity of fewer channels and the fees are reduced.}
    \label{fig:transAggrIntro}
\end{wrapfigure}

\subsection{Organisation}
The remainder of this paper is organized as follows.
In Section \ref{sec:modelandprotocol}, we present the model and additional assumptions, as well as the proposed skeleton protocol for solving the transaction aggregation problem and the desired properties it should fulfil.
In Section \ref{sec:protocol}, we present our implementation of the skeleton protocol, namely \sys,  and in Section \ref{sec:analysis} we prove that our implementation provides the desired guarantees.
In Section \ref{sec:discussion} we demonstrate the connection to the netting problem and discuss further extensions of our protocol. Finally, in Section~\ref{sec:related} we compare to related work and we conclude with Section~\ref{sec:conclusion}.

\section{Model and Protocol Overview}
\label{sec:modelandprotocol}
In this section, we first describe payment channel networks, and then introduce the assumptions of our work.
Later, we formulate transaction aggregation  as an optimization problem, then provide a high-level overview of our solution (skeleton protocol), and conclude the section by defining the desired properties for our system.

\subsection{Payment channel networks}\label{subsec:pcn}
In this subsection, we provide a brief introduction to payment channel networks with standard fees as they currently operate.

We assume that a payment channel is characterized by the public keys of the two users opening it, and the funds and forwarding fees that each user adds to each direction of the channel upon its creation \cite{poon2015lightning}. 
Let $u,v$ be two users that have opened a payment channel $\{(u,v), (v,u)\}$ and $c(e)$, $\basefee(e)$, $\propfee(e)$,  $e\in \{(u,v), (v,u)\}$ be the current capacity, the base forwarding fee, and the proportional forwarding fee of each channel direction, respectively.
Users $u$ and $v$ can do an arbitrary number of payments to each other by subtracting any non-negative amount $x\leq c(e)$,  $e\in \{(u,v), (v,u)\}$, and increasing $c(e')$ by $x$, where  $e'\in \{(u,v), (v,u)\}\setminus\set{e}$.
Note that the forwarding fees and $c(u,v) + c(v,u)$ is public information, while each summand (current balance) is not.

A payment channel network (PCN) is the network built over all payment channels. 
Users that do not share payment channels can perform payments via paths in the network, by covering a forwarding cost to the intermediate nodes of the path.
For example, if payment channels exist between $u,v$ and $v,w$, but not between $u,w$, then $u$ can transfer an amount $x$ to $w$ in two steps: (i) $u$ transfers $x + (\basefee(v,w) + \propfee(v,w)\cdot x)$ coins to $v$ and 
(ii) $v$ transfers $x$ coins to $w$, i.e. $\basefee(v,w) + \propfee(v,w)\cdot x$ are the forwarding fees charged by $v$. 
The transaction is feasible only if $x + (\basefee(v,w) + \propfee(v,w)\cdot x) \leq c(u,v)$ and $x\leq c(v,w)$. 

In general, suppose that user $s$ wants to pay $x$ coins to user $r$, via a payment path $\{(s, u_1), ..., (u_k, r) \}$. Then $u_k$ receives $rcv_k = x + (\basefee(u_k,r) + \propfee(u_k,r)\cdot x)$ and forwards $x$ coins to $r$, $u_{k-1}$ receives $rcv_{k-1} = x + (\basefee(u_{k-1}, u_k) + \propfee(u_{k-1}, u_k)\cdot rcv_k)$ and forwards $rcv_k$ coins to $u_k$, and so on.
Moreover, the transaction is feasible only if every node in the path has sufficient capacity to route the forwarded amount.
This is not always the case, since current balances are private and senders might have to try different payment paths towards a successful transaction.

A transaction along a payment path is successful when all intermediate transactions go through, and otherwise none of them is performed. That is, PCNs must ensure atomicity in transaction routing. In the Bitcoin Lightning Network, this is ensured by Hash Timelock Contracts (HTLCs) \cite{poon2015lightning}.

\begin{figure}
    \centering
    \includegraphics[scale=0.4]{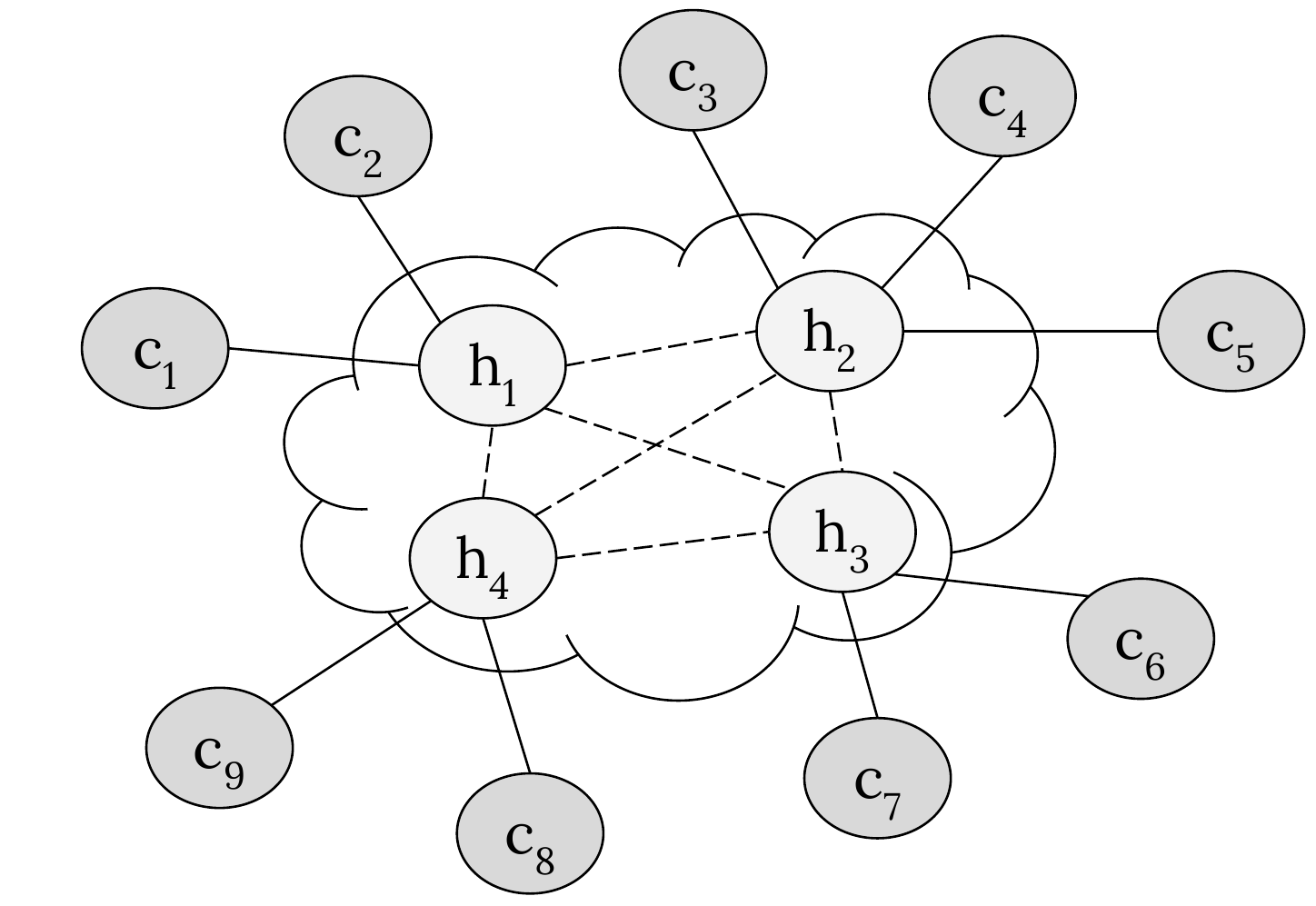}
    \caption{Example of the restricted PCN topology. Each client (nodes $\{c_1, \ldots, c_9\}$) is connected to one of hub nodes (nodes $\{h_1, \ldots, h_4\}$) and the hub nodes are well interconnected (potential edges are dashed). We propose a solution for hub connectivity in our protocol description.}
    \label{fig:PCNtopology}
\end{figure}

\subsection{Assumptions}
\label{sec:additional assumptions}
In this subsection, we introduce the cryptographic and network assumptions as well as the considered PCN topology.

\paragraph{Cryptographic assumptions.}
We assume the existence of secure communication channels between users.

\paragraph{Blockchain and network model.}
We assume the underlying block-chain satisfies persistence and liveness as defined in \cite{garay2017bitcoin}. We also assume a synchronous network model. That is, there is a known network delay which bounds the time needed for any user to receive any incoming message. 

\paragraph{Studied PCN topology} We restrict the PCN topology to allow for a computationally tractable implementation of transaction aggregation. Let $V = H\cup C$ be the set of all users, which we split to a relatively small set of hub nodes $H$ and the remainder set of client nodes $C$.
The hubs are well interconnected, and the client nodes are those connected to only one hub (Figure \ref{fig:PCNtopology}). 

\paragraph{Users and transactions.}
We assume the users of the protocol follow the protocol specification unless they can monetarily gain from deviating from it.

Over time, users accumulate transactions before initiating the protocol. The list $\CT_u$ of transactions that user $u$ submits must not exceed the capacity of the channel $u$ has with their hub $h_u$. In other words, users must not authorize a total payment of funds exceeding their balance with the hub. Moreover, the same must also hold for the set of incoming transactions where $u$ is the recipient.

\paragraph{Fee function.}
We assume that the forwarding fees charged by hubs for each channel are publicly known before the protocol execution. We also assume that the forwarding fee is a non-negative function $F(a,b)$ of the initial and final states $a, b$ of the channel and that the triangle inequality holds: for all states $a, b, c$: $F(a,c) \leq F(a,b) + F(b,c)$. As described in Section~\ref{subsec:pcn} above, the standard fee functions are simply a base fee plus a proportional fee, which do satisfy our more general assumptions. Intuitively, it should not be cheaper to forward a given amount along a channel in multiple parts compared to one.

Note that our proposed solution does not aim to minimize the fees, but rather to maximize the volume of transactions that are cleared. We nevertheless show that transaction aggregation implicitly reduces incurred fees for the above class of fee functions. We discuss a more sophisticated consideration of fees in Section~\ref{sec:discussion}.

\subsection{(Computational) Problem Definition}
\label{sec:compProbDef}

Let $G=(V,E)$ be a directed graph that models a PCN with integral edge capacities $c(e)$. Let $V = \set{v_1, \ldots v_n}, \len{E}=m$. 
A payment channel between $v_i$ and $v_j$ is modelled as two unidirectional edges $(v_i, v_j)$ and $(v_j, v_i)$.

A flow through this network is an $m$-dimensional vector $\vf = (\vf(e))_{e \in E}$ that satisfies the capacity constraints:
$$0 \leq \vf(e) \leq c(e) \quad \forall e \in E.$$

A demand vector $\vd = (d_1, \ldots d_n) \in \mathbb{Z}^n$ is a vector representing the monetary flow through each vertex.
For instance, a transaction of $1$ unit from $v_1$ to $v_2$ can be represented as the demand vector $(1, -1, 0, \ldots 0)$. 
A flow $\vf$ is said to route a demand vector $\vd$ if for every vertex $v$:
$$\sum\limits_{(v,u)\in E} \vf(v,u) - \sum\limits_{(u,v)\in E} \vf(u,v) = \vd(v)$$

A demand vector is feasible if there is a flow that routes it. For example, the demand vector $(1,-1,0,\ldots 0)$ representing a transaction of amount $1$ from $v_1$ to $v_2$ is feasible if, say, there is an edge (payment channel) from $v_1$ to $v_2$ of capacity at least $1$.

Consider a list of transactions $\CT = [\vt_1, \ldots \vt_k]$ each represented as a demand vector, i.e., $\vt_i$ is a demand vector with $w_i$ in the position of the sender, $-w_i$ in the position of the receiver, and zero otherwise, where $w_i$ is the transaction amount of $\vt_i$.
We denote by $|\vt_i| := w_i$ the amount of transaction $\vt_i$.
We define \textit{transaction aggregation} as the process of:\\
(i) algebraically adding up all the demand vectors in $\CT$, to obtain a demand vector $\vd = \sum\limits_{\vt \in \CT} \vt$,\\ 
(ii) computing a flow $\vf$ that routes $\vd$, and then\\ 
(iii) executing this flow on the payment channel network as a means of executing all the transactions in $\CT$.

\begin{wrapfigure}[20]{r}{0.5\textwidth}
    \centering
    \includegraphics[width = 0.4
    \textwidth]{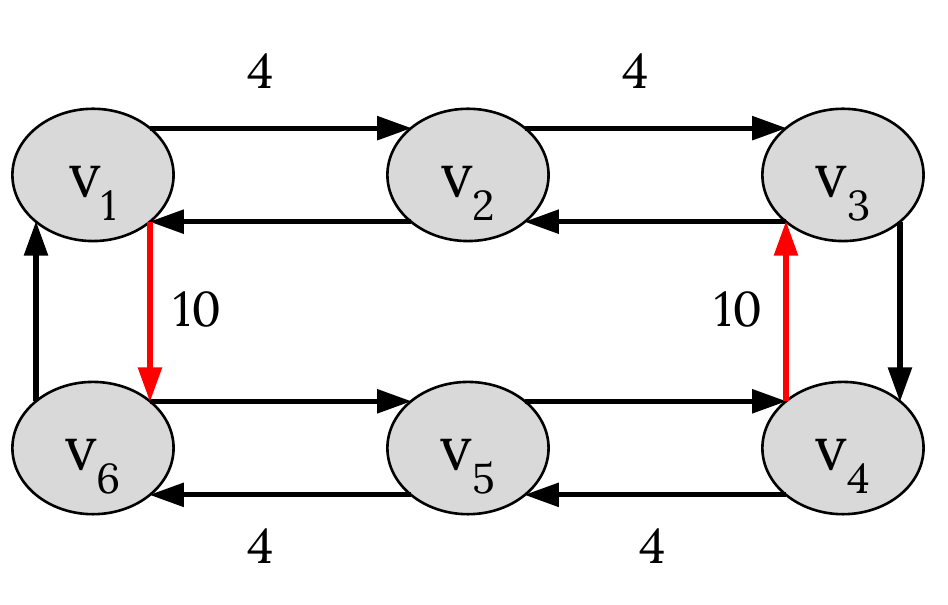}
    \caption{A transaction aggregation example. Non-depleted edges are annotated with their capacity (e.g. $c(v_2, v_1) = 0$, while $c(v_1, v_2) = 4$). Let $\CT = [(10,0,-10,0,0,0), (10,0,-10,0,0,0), (0,0,0,10,0,-10)]$, i.e., two transactions of 10 coins from $v_1$ to $v_3$ and one transaction of 10 coins from $v_4$ to $v_6$. These transactions are infeasible if executed sequentially. However, $\CT^* = [(10,0,-10,0,0,0), (0,0,0,10,0,-10)]$ can be aggregated to the demand vector $\vd = (10, 0, -10, 10, 0, -10)$, and executed by the flow $\vf$ with $\vf(v_1, v_6)= \vf(v_4, v_3)=10$ and $\vf(e)=0$ otherwise.}
    \label{fig:aggregationExamples}
\end{wrapfigure}

The computational challenge of transaction aggregation lies in finding a subset $\CT'$ whose aggregate demand vector can feasibly be routed, which we define below. 

\begin{definition}
\label{def:transactionAggr}
The Transaction Aggregation Problem on a directed graph $G$ concerns computing a sublist $\CT^* \subseteq \CT = \left[\vt_1, \ldots \vt_k\right]$ that can feasibly be aggregated, and which is optimal in terms of total throughput $\sum_{\vt_i \in \CT^*} |\vt_i|$.
\end{definition}

Next, we reformulate the Bank Clearing Problem~\cite{GuntzerJL98,ShafranskyD06} in our notation, noting it to be a special case of transaction aggregation. 

\begin{definition}
\label{def:bcp}
The Bank Clearing Problem with $n$ participants $\set{v_1, \ldots, v_n}$, each with current balance (also known as cover money) $\vb = (b_1, \ldots b_n)$ concerns computing a sublist $\CT^* \subseteq \CT = \left[\vt_1, \ldots \vt_k\right]$ that can be aggregated into a demand vector $\vd^*$ such that $\vd^* \leq \vb$, and which is optimal in terms of total throughput (also known as clearing volume) $\sum_{\vt_i \in \CT^*} |\vt_i|$.

\end{definition}

Let us briefly elaborate on the throughput benefits of transaction aggregation. First, it may be that the complete list $\CT$ cannot be aggregated into a feasible demand vector, but sublists of transactions can (Figure \ref{fig:aggregationExamples}). Figure \ref{fig:aggregationExamples} also illustrates that routing an aggregate may require less flow along channels, or use shorter paths, and thus incur lower fees.
Secondly, there is a varying degree of cancelling out of the demand vector $\vd$, depending on transactions in $\CT$. That is, transactions might be completely independent (payment paths are disjoint), or all of them can cancel out (e.g. $s$ and $r$ want to pay each other the same amount), leading to a smaller demand vector, or anything between these two extremes. This form of canceling out is beneficial regardless of channel capacities. We give an example in Figure \ref{fig:cancelout}.
Note that the cancelling out feature of transaction aggregation might render a list of transactions feasible, even though channel capacities do not suffice for sequential execution. 

In what follows, we will refer to oracles $\CO$ as algorithms that can solve the transaction aggregation problem, possibly under certain assumptions on $G$ and $\CT$, to return a feasible sublist $\CT^*$ that maximizes a linear function of $\CT$.
The transaction aggregation problem is generally computationally hard as we show in Section~\ref{sec:analysis}. 
\begin{wrapfigure}[24]{l}{0.5\textwidth}
    \centering
    \includegraphics[width=0.48\textwidth]{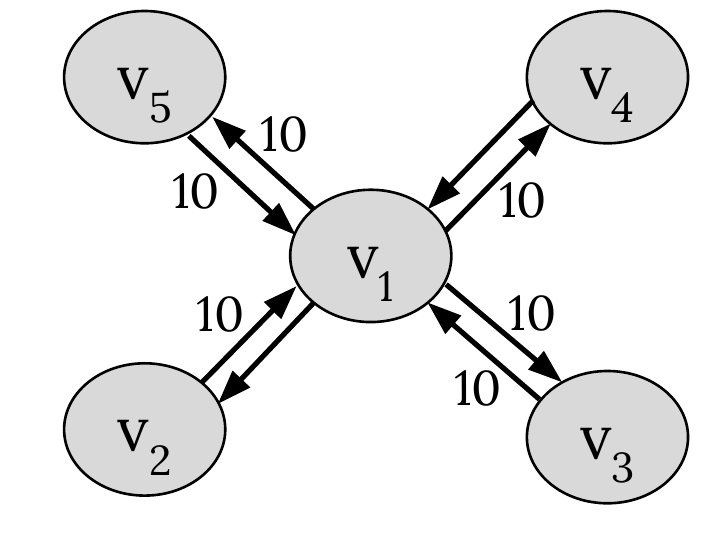}
    \caption{Examples of cancelling out of transactions. Capacities are annotated only on the edges relevant to the examples. We consider two examples: $\CT_1 = [(0,5,0,-5,0), (0, 0, 5, 0, -5)]$ and $\CT_2 = [(0,0,5,0,-5), (0,0,-2,0,2)]$. Transaction aggregation benefits $\CT_2$ by reducing total flow but not $\CT_1$. That is, aggregating $\CT_2$ produces a smaller demand vector $(0,0,3,0,-3)$, while aggregating $\CT_1$ yields $(0,5,5,-5,-5),$ which is equivalent to sequential execution.}
    \label{fig:cancelout}
\end{wrapfigure}

\subsection{Protocol overview}
\label{subsec:overview}

\begin{figure*}[t]
    \centering
    \includegraphics[scale=0.45]{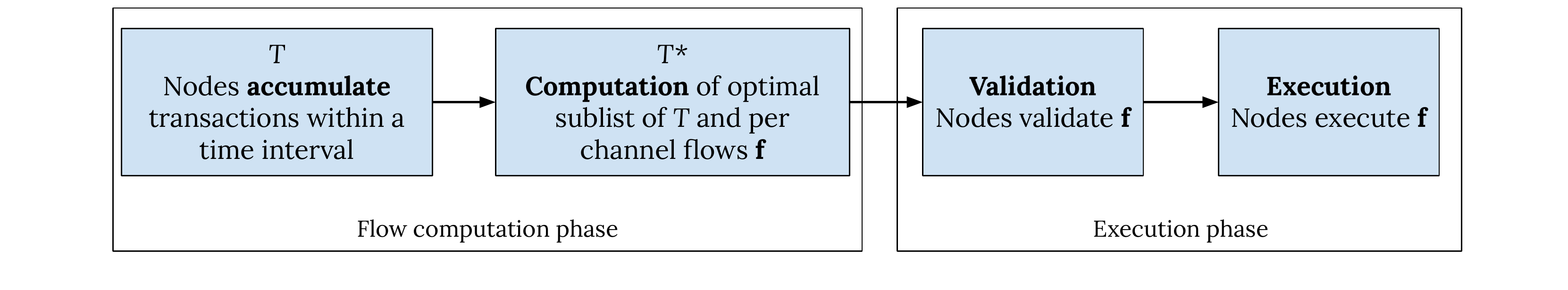}
    \caption{Skeleton protocol for transaction aggregation.}
    \label{fig:skelprot}
\end{figure*}

In this section, we introduce a protocol abstraction or the ``skeleton'' protocol for solving and executing the transaction aggregation problem in a payment channel network formulated as described in Section~\ref{sec:compProbDef} under our assumptions.

Our skeleton protocol consists of two phases: a flow computation phase and an execution phase. The goal of the flow computation phase is to privately compute and output an optimal transaction flow from the list of transaction requests submitted to the protocol. The goal of the execution phase is to execute the flow in an atomic fashion. 
An illustration of the flow of the skeleton protocol can be found in Figure~\ref{fig:skelprot}.

\paragraph{Flow computation phase} 
The flow computation phase initiates when a pre-specified amount of time elapses or when  sufficiently many transactions have been collected by the users\footnote{The decision of when the protocol triggers is left to the users. The core idea is to be able to accumulate many transactions per user may that be after some time interval or when a specific number of transactions per user is reached.}. 
In this phase, the list of transactions $\CT$ as well as the respective channel balances are taken as input. The output of the flow computation phase is a sublist of transactions $\CT^*$ that will be eventually executed in the PCN, along with a resultant flow $\vf$ that routes the aggregate of these transactions.

\paragraph{Execution phase}
Out of the computed solution $\CT^*$ and $\vf$, individual users $v$ are provided with the sublist of $\CT^*$ containing only those transactions where $v$ is a sender or recipient, and the flow $\vf(e)$ for edges $e$ adjacent to $v$. With this data, each user verifies the validity of the solution, and that balance security holds, i.e., the user does not lose any money given the resulting flow and the transactions to be executed.
After this verification, the users must execute $\vf$ atomically. The flow output from the previous flow computation phase might involve disconnected components of the network; therefore, simple HTLCs may not suffice to atomically execute the transactions $\CT^*$. Instead, a clever solution is needed where atomicity of all transactions in $\CT^*$ is maintained regardless of the graph topology. 

We will introduce the specific protocols that will be leveraged to implement the skeleton protocol described above in Section~\ref{sec:protocol}.

\subsection{Desired Properties}\label{subsec:properties}
In this section, we discuss the desiderata of our protocol \sys and then formally define the protocol goals.

First, we demand that our protocol is practical, i.e., that the solution of the transaction aggregation problem is computationally feasible.
Second, \sys should maintain security of the channel balances, meaning that the users should not lose money due to the protocol execution. In particular, the intermediaries (hubs) should not lose any money, while the users should only have a deficit equal to the sum of the amounts of the selected transactions in $\CT^*$. This security property is encapsulated by \textit{Balance Security}.

Furthermore, we demand two optimization properties: one that ensures that our solution transfers the maximum amount of value possible, termed \textit{Optimality}, and one that ensures our solution is cheaper (or equivalent at worse) than executing the same transactions sequentially in the PCN, termed \textit{Cost Efficiency}. The former property guarantees that \sys does not output a trivial solution and the solution maximizes the transaction throughput, while the latter guarantees that \sys is beneficial in terms of fees compared to the no-aggregation solution where each user executes its transactions sequentially and without the need for coordination on the PCN. 

Lastly, due to the potential aggregation of transactions per user, the resulting flow may enhance the privacy of the protocol compared to the simple sequential execution per user. 
For instance, if the transaction aggregation selects two transactions for a user that partially cancel out (e.g., the user receives 5 coins and sends 3 coins), the user's channel to the hub will only be updated by the difference (e.g., 5-3=2 coins); thus, the hub only learns the aggregated flow and not the individual transactions. 
To capture this improvement on privacy guarantees, we define the \textit{Privacy} property. 

In particular, we define the notion of privacy we wish to achieve during the entire course of our protocol, i.e.\ both the flow computation and execution phases. 
We want to enforce, firstly, that each uninvolved user should only learn that they are not involved in the protocol. 
Secondly, each involved user (whether client or hub) should only learn the flow output on each of the incident channels to the hubs.
We stress that the above two conditions imply value privacy: users will not know the flow (and thus the transaction amounts) along channels they are not a part of. 
Finally, we allow each involved user to learn the number of other involved users in the protocol \footnote{Ideally, involved users should not learn any information about the set of other involved users in the protocol. However, our proposed flow execution protocol requires the number of involved users to be known in advance. We are working on future work to bound the number of involved users each involved user needs to know.}. 

Formally, \sys should achieve the following properties:

\begin{definition}[Computational Feasibility]\label{def:efficient}
The transaction aggregation problem as per Definition \ref{def:transactionAggr} is fixed-parameter linear (\textsc{FPL} complexity class) \cite{downey2012parameterized}, i.e., its running time is polynomial in the number of clients and exponential in the number of hubs.
\end{definition}

\begin{definition}[Balance Security]\label{def:balsec}
The change in user balances, after the execution of our protocol, must be the aggregated demand vector of some sublist $\CT' \subset \CT$ of transactions. 
\end{definition}

\begin{definition}[Optimality]\label{def:optimality}
For a given list of transactions, the aggregation should select a maximal sublist in the sense of total demand fulfilled.
\end{definition}

\begin{definition}[Cost Efficiency]\label{def:costef}
\label{def:cost}

The total fees levied to users should be less than or equal to the amount that would have been levied had the transactions been processed sequentially.
\end{definition}

\begin{definition}[Privacy]\label{def:privacy}
Our notion of privacy is based on the following indistinguishability game between an passive adversary $\mathcal{A}$ and a transaction aggregation and execution protocol $\Pi$ when run on the subgraph of users interested 
in participating in the protocol $G=(V,E)$:

\begin{figure*}
    \centering
    \includegraphics[scale=0.45]{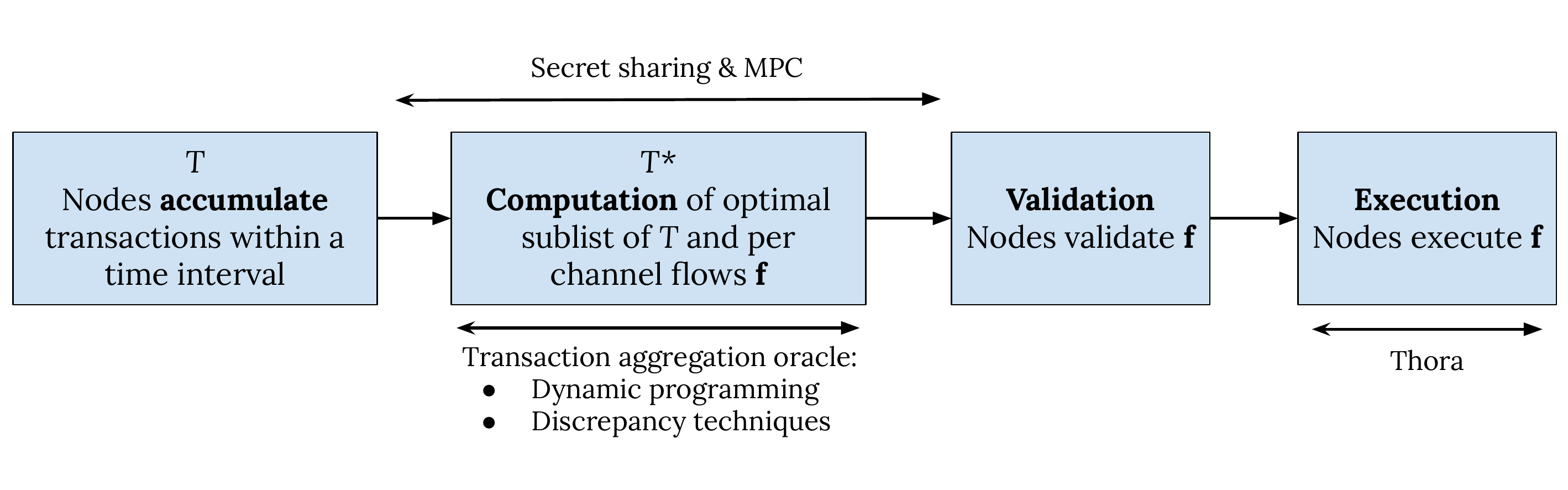}
    \vspace{-0.5cm}
    \caption{Proposed implementation of the skeleton protocol for transaction aggregation. Secret sharing is used before the computation module to select the delegates and share the inputs, MPC is used throughout the computation module and until the module's output is delivered to the nodes. The tools used for the computation and (atomic) execution modules are noted below those boxes.}
    \label{fig:skelprotimpl}
\end{figure*}

\begin{itemize}
    \item $\mathcal{A}$ chooses a subset of users $V^{\mathcal{A}} \subset V$ to corrupt, i.e., $\mathcal{A}$ gains access to the transcripts of every corrupted user $v \in V^{\mathcal{A}}$. 
    $V^{\mathcal{A}}$ can consist solely of hubs, clients, or a mix of both.
    Let $G^{\mathcal{A}}$ be the subgraph created by taking the union of all corrupted nodes and their incident edges. 
    Additionally, if a hub node is in $V^{\mathcal{A}}$, we add all the other hubs in the network and all their incident edges to $G^{\mathcal{A}}$.
    \item $\mathcal{A}$ chooses the following for $i \in \{0,1\}$: $\mathcal{A}$ creates a list of transactions tuples $\mathcal{T}^i = \left[(x_j, s_j, r_j)\right]_{j=1}^{n_i}$ where the $j$th transaction tuple consists of a transaction amount $x_j$, the vertex $s_j$ that requested the transaction, and the vertex $r_j$ that should receive the transaction. $n_i$ is the number of transaction tuples in the $i$th list.
    \item For every node and channel in $G^{\mathcal{A}}$, the following condition needs to hold: the resulting flow returned by $\Pi$ restricted to the subgraph $G^{\mathcal{A}}$ when run on $G$ and when given transaction tuple lists $\mathcal{T}^0$ and $\mathcal{T}^1$ has to be the same. 
    Additionally, the set of involved users as computed by $\Pi$ should be the same when given $\mathcal{T}^0$ or $\mathcal{T}^1$ as input. 
    \item (Challenge phase.) We choose a random bit $b \in \{0,1\}$ and run $\Pi$ with transaction tuple list $\mathcal{T}^b$.
    \item $\mathcal{A}$ gets the flow output sent to each corrupted user from $\Pi$.
    \item $\mathcal{A}$ outputs a bit $b'$. If $b'=b$, $\mathcal{A}$ wins the game.
\end{itemize}

We say a transaction aggregation and execution protocol $\Pi$ is $\epsilon$-private if $\mathcal{A}$ wins the above game with
probability at most $\frac{1}{2} + \epsilon$, and private if it is $\epsilon$-private for some negligible $\epsilon$.
\end{definition}

\section{\sys Architecture}\label{sec:protocol}
In this section, we present \sys, a protocol that solves the transaction aggregation problem under our system model. 
In Section \ref{sec:WiserTopology}, we present the topology over which \sys is designed to operate. This restriction is necessary to design a computationally tractable solution to the otherwise \NP-hard transaction aggregation problem (\NP-hardness is proved in Section \ref{sec:oracle}).
Then we present \sys, which is illustrated in an abstract level in Figure \ref{fig:skelprotimpl}. 

\sys consists of a flow computation phase (Section \ref{sec:computation}) and an execution phase (Section \ref{sec:execution}).
Nodes accumulate transactions within a fixed time interval and then secret share these transactions and their balances to all hub nodes, which in turn run a flow computation module. 
This module privately computes an optimal (in terms of maximum throughput) sublist of transactions to be executed and a flow that realizes these transactions.
\sys then proceeds to the execution phase. Each node locally validates the flow and selected transactions received by the computation module and then proceeds to atomically executing the flow. That is, either all flows are executed or none.
We summarize \sys in Algorithm \ref{alg:traggr} (Section \ref{sec:algdescr}).


\subsection{PCN topology \& Channel Factories}
\label{sec:WiserTopology}
We consider the restricted topology presented in Section \ref{sec:additional assumptions}. That is, the set of nodes is composed by a small set of hub nodes $H$ and a set client nodes, such that each client is connected to exactly one hub. 
This assumption ensures that transactions, and general demand vectors, can only be routed by unique flows in our PCN. In such a PCN, transaction aggregation problem is fixed parameter tractable (see \Cref{sec:oracle} for more details).
Moreover, we assume that the hub nodes are interconnected with a channel factory. 

Channel factories were first introduced as an intermediate construction between layer $1$ on-chain transactions and layer $2$ payment channel transactions\cite{burchert2017scalable}.
These factories are set up between $3$ or more individuals, and henceforth can be used to instantly open and close off-chain standard payment channels between any $2$ of the involved parties.

Two parties set up a conventional payment channel by committing some amount of their individual funds into a $2$-signature wallet considered as a shared account. In much the same way, $k$ parties may set up a channel factory among themselves by locking funds in a $k$-signature wallet. The internal working of a channel factory relies in fact on conventional payment channels, however these are not broadcast publicly nor included in the blockchain. Since the funds in the internal channels of the factory can be reallocated to new channels instantly, channel factories provide greater flexibility and scalability. 

A conventional payment channel can execute arbitrarily many transactions between two parties provided the resultant balance lies within the channel's capacity, and a channel factory can achieve the same functionality for $3$ or more parties. Formally, a channel factory between $r$ parties $v_1, \ldots v_r$ of current capacity $(c_1, c_2, \ldots c_r)$ (the total capacity is $C = c_1 + c_2 + \ldots c_r$ and is fixed) can route any demand vector $(d_1, \ldots d_r)$ provided $d_i \leq c_i, \forall 1 \leq i \leq r$. Another equivalent description of channel factories is that the transaction aggregation problem on $r$ parties connected by a channel factory reduces to the bank clearing problem on $r$ banks.

There is a natural generalization of fees from payment channels to channel factories, where a node $v_i$ would charge a fee for a certain state transition from $\vc = (c_1, c_2, \ldots c_r)$ to $\vc'= (c_1', c_2', \ldots c_r')$ if $c_i'< c_i$ ie. if their balance decreases. Our assumption on the fee function as stated in Section~\ref{sec:additional assumptions} generalized to the following: the total forwarding fee must be a non-negative function $F(\vc_0, \vc_1)$ of the initial and final states $\vc_0,\vc_1$ of the channel factory. $F$ must also satisfy the triangle inequality: for any states $\vc_0, \vc_1,\vc_2$, we have
$$F(\vc_0, \vc_2) \leq F(\vc_0,\vc_1) + F(\vc_1, \vc_2).$$

The improvements of channel factories come at the cost of minor liveness assumptions, namely, that the internal channels in a factory can only be reorganized when all involved parties are online. Although this assumption is generally impractical, it is a reasonable one to make in the setting of \sys. Only the hubs, and not the clients, are required to set up a channel factory. Based on their role as payment service providers, it is fair to expect them to be online regularly.


Channel factories are not crucial to our protocol, but they provide an increase in throughput compared to solely using payment channels. We recommend a channel factory between the hubs, as opposed to the hubs building various channels between each other for the following reason: suppose the hubs commit the same amount of capital to a set of payment channels between them to create a PCN $G_H$ between them. Any demand vector $\vd$ that can feasibly be routed in $G_H$ can also be routed by the channel factory. On the other hand, there are always demands $\vd$ that a given $G_H$ cannot route but the factory can. In this way, the throughput of the network is maximized for a given amount of capital per hub.

\begin{algorithm*}[t]
\caption{\sys~--~pseudocode for node $u$}
\label{alg:traggr}
\LinesNumberedHidden
\KwInput{$out(\CT,u)$
a set of outgoing transactions from node $u$, where $\CT := \cup_u out(\CT,u)$ and the balances of adjacent channels $\{c(u,w), c(w,u) \,|\, \text{$u$ and $w$ share a channel}\}$.}
\KwOutput{Result of atomic execution of a flow that realizes a sublist $out(\CT^*,u) \subseteq out(\CT,u)$, such that $\CT^* = \cup_u out(\CT^*,u)$ is optimal in terms of total throughput $\sum_{\vt\in\CT} |\vt|$.}
 \textbf{Macros:}\\ 
$\vf(u)$: denotes the entries of a flow $\vf$ referring to edges adjacent to $u$\;
$in(\CT,u)$, $out(\CT,u)$: the list of incoming, and respectively, outgoing transactions with respect to node $u$ in the transaction list $\CT$\;
\LinesNumbered
\setcounter{AlgoLine}{0}

\tcc{Flow computation phase}
\textit{secretSharingProtocol}$(out(\CT,u), \{c(u,w), c(w,u) \,|\, \text{ $u$ and $w$ share a channel}\})$\label{algl:secretsharing}\tcc*{share input to hubs}

\If{$u$ is a hub (i.e., MPC delegate)\label{algl:ifhub}}{
\tcc{Input validation}
\For{$v\in V\setminus H$\label{algl:inputvalidationstart}}{
\mylIf{$\sum_{\vt\in out(\CT,v)} |\vt| \geq c(v,h_v) $}{$\CT \leftarrow \CT \setminus out(\CT,v)$;\hfill\texttt{/* exclude txns violating outgoing capacity  */}}

\mylIf{$ \sum_{\vt\in in(\CT,v)} |\vt| \geq c(h_v,v)$}{ $ \CT \leftarrow \CT \setminus in(\CT,v)$; \label{algl:inputvalidationend}\hfill\texttt{/* exclude txns violating incoming capacity */}}
}
\textit{optimizationOracle()}\label{algl:oracle}\;
\Return $out(\CT^*,v)$, $\vf^*(v)$ to each node $v$\label{algl:MPCreturn}\;
}
 $out(\CT^*,u), \vf^*(u) \gets MPC()$\label{algl:MPCoutput}\tcc*{$u$ receives MPC output from hubs}

\tcc{Execution phase}
\mylIf{$\exists\vt \in out(\CT^*,u)$ s.t. $\vt \notin out(\CT,u)$}{abort; \hfill\texttt{/* validation of $out(\CT^*,u)$ */}\label{algl:validationstart}}

\mylIf{$\exists\vt \in in(\CT^*,u)\cup out(\CT^*,u)$ s.t. $\vt \notin in(\CT^*,v)\cup out(\CT^*,v) \land \vt$ is a transaction between $u$ and $v$}{abort;} 

\mylIf{$\vf^*(u)$ does not match the entries for $u$ in $\vf^*(v)$ s.t. transactions between $u$ and $v$ appear in $\CT^*$}{abort; }

\mylIf{$\sum\limits_{e=(v,u) \in E}\vf(e) - \sum\limits_{e=(u,v)\in E}\vf(e) \neq \vd(v)$}{abort \hfill \texttt{/* validation of $\vd$ vs $\vf$ */}\label{algl:validationend}}

$Thora(\vf^*(u))$\label{algl:thora}\tcc*{$u$ executes $\vf(u)$ by participating in Thora \cite{aumayr2022thora}}
\end{algorithm*}

\subsection{Flow computation phase}
\label{sec:computation}

The flow computation phase begins when sufficiently many transactions have been collected by the users, or when a pre-specified time has elapsed. 
In this phase, the transaction aggregation problem is solved. 
In \sys we aim to preserve the privacy of transactions. For this reason, we employ secret sharing and secure multi-party computation to solve the optimization problem.
The output of the flow computation phase is a list of transactions  $\CT^*$ and a flow $\vf$ to be executed in the next phase.
The two separate components of this phase are discussed below.

\paragraph{Secret Sharing the Transaction Aggregation Input}
With secret sharing we refer to protocols that allow a user to distribute a secret among a group of participants, each of whom is allocated a share of the secret. The secret can be then reconstructed but only when a sufficient number of shares are combined (threshold), meaning that a sufficient number of participants cooperate. Furthermore, the individual shares (or a subset of the shares with cardinality less than the threshold)  leak no information on the secret. 
When $n$ is the number of participants that get shares and $t$ is the reconstruction threshold, the protocol is called $(t,n)-$threshold secret sharing. 
\sys is agnostic to the selection of a specific secret sharing protocol, so any such scheme can be used such as Shamir's secret sharing~\cite{shamir1979how}.

In \sys, the list of transactions $\CT$ as well as the channel capacities in either direction are secret shared by the users, in order to perform a multi-party computation of the transaction aggregation problem. 
In particular, we assume the protocol specifies a specific block header from the underlying blockchain that provides common randomness to the users. 
Using this randomness, the users select $k$ delegates among the hubs, which will be the ones that will later execute the optimization oracle to provide a solution of the transaction aggregation problem to all users in a privacy-preserving manner.
The users then $(k,k)$-threshold secret share to the $k$ delegates their inputs, specifically their outgoing transactions and their channel balances. 
For privacy reasons, we demand that all users send transactions -- even if they are of zero value, when a user has only incoming transactions.

We stress that the $k$ delegates should be sampled from the set of hubs (as opposed to sampling from clients or a mixture of both) for the following reasons: 
firstly, as hubs are large business holders in the PCN, it becomes expensive, hence difficult, for an adversary to launch Sybil attacks to counterfeit them in order to control a disproportionate fraction of the delegates. 
Secondly, computing the solution to the transaction aggregation problem imposes resource and liveness constraints on the delegates: all delegates need to have sufficient computational power to perform the computation, and they need to be online throughout the computation process. 
It would therefore be more reasonable to assume that hubs, as large financial service providers in the PCN, indeed have the capacity to fulfil these requirements.


\paragraph{Input Validation} Once the transactions and constraints have been secret shared, the delegates should validate the well-posedness of the input. In particular, the delegates verify that no user submits outgoing transactions that would exceed the capacity of their channel with their hub, as assumed in Section~\ref{sec:additional assumptions}. 
 
\paragraph{Solving the Transaction Aggregation Problem}
Let $\CO$ denote an optimization oracle that, when given a PCN graph $G$ and a list of transactions $\CT$, returns a sublist $\CT^*$ along with a resultant flow $\vf$ that routes the aggregate of $\CT^*$. $\CT^*$ is optimal in the sense that the throughput $\sum_{\vt_i \in \CT^*}\limits |\vt_i|$ is maximized.
We elaborate on our proposed optimization oracle below.

As stated in Section \ref{sec:WiserTopology}, our PCN consists of hubs, connected by a channel factory, and clients connected to one hub only. The assumption on $\CT$ is that for each client $v$, their channel with the hub contains enough liquidity to route all transactions in $\CT$ where $v$ is the sender. The same channel must also have sufficient liquidity in the opposite direction to route all transactions where $v$ is the recipient.

Under these assumptions, the transaction aggregation problem reduces to an integer program with $\len{\CT}$ variables and $\len{H}$ constraints. Although solving such integer programs generally takes time exponential in $\len{\CT}$, we employ a recent result by Eisenbrand and Weismantel~\cite{eisenbrand} that solves integer programs in time linear in $\len{\CT}$, albeit exponential in $\len{H}$, i.e., the algorithm is fixed-parameter linear \cite{downey2012parameterized}. Their algorithm uses discrepancy techniques (which address the question of rearranging a sequence of vectors to limit the norm of their partial sums) to facilitate dynamic programming.
In Theorem \ref{thm:txagrr} of Section \ref{sec:oracle} we prove that the proposed optimization oracle solves the transaction aggregation problem and compute its computational complexity.

\subsection{Execution phase}
\label{sec:execution}

The flow computation phase concludes with a solution to the transaction aggregation problem, albeit encoded as secret shares. The solution consists of the list of accepted transactions $\CT^* \subseteq \CT$, along with a resultant flow $\vf$ on the network. An execution of this flow would be equivalent to a simultaneous execution of every transaction in $\CT^*$. However, the solution must first be checked for validity, and then be executed atomically though the entire network.
In the following, we discuss how the validation of the solution will be performed as well as the execution of the accepted transactions through the network in an atomic manner. 
We further propose a specific structure for the interconnection of the hubs, known as payment channel factories~\cite{burchert2017scalable}.

\paragraph{Flow Validation}
To check that a solution $(\CT^*, \vf)$ is valid is to check that $\CT^* \subseteq \CT$, that aggregation of $\CT^*$ yields a feasible demand vector $\vd = \sum\limits_{\vt \in \CT^*} t$, and that the flow $\vf$ routes $\vd$.

A trusted third party could easily verify the above, but we require users to verify these conditions \textit{locally}. We present a method for local verification, such that a solution is valid if and only if every user locally verifies the above conditions.

A user $v$ is supplied with a list $\CT^*_v \subset \CT^*$ such that a transaction $t \in \CT^*_v$ whenever $v$ is the sender or the recipient. Representing transactions as demand vectors, $\CT^*_v = \set{\vt \in \CT^*: \vt_v \neq 0}$. Moreover, $v$ is supplied with $\vf(e)$ for every edge $e$ adjacent to $v$ ie. $e=(u,v)$ or $e=(v,u)$.

The users can verify that the provided local data is consistent with that of other users: for every $t \in \CT$, the sender-recipient pair $u, v$ can verify that $t \in \CT^*_u$ if and only if $t \in \CT^*_v$. Similarly, every pair of adjacent vertices $u,v$ with $(u,v)\in E$ can verify they have both been provided the same value of $\vf(u,v)$.

$v$ verifies that every transaction in $\CT^*_v$ also belongs in $\CT_v$ and therefore in $\CT$. $v$ computes their aggregate demand $\vd(v)$ under $\CT^*_v$. To verify that $\vf$ routes $\CT^*$ locally, $v$ simply checks if the net flow of $\vf$ through $v$ is equal to $\vd(v)$:
$$\sum\limits_{e=(v,u) \in E}\vf(e) - \sum\limits_{e=(u,v)\in E}\vf(e) = \vd(v).$$

\paragraph{Atomic Execution}
After this verification, the users must execute $\vf$ atomically, i.e., 
ensuring that our protocol executes the set of payments atomically and efficiently after performing transaction aggregation, is critical for  security.
An incomplete update of channels does not typically correspond to any meaningful financial activity like the execution of a sublist of aggregated transactions. The flow output from the previous flow computation phase might involve disconnected components of the network. To ensure atomicity of channel updates across all disjoint paths, we employ the solution proposed in Thora \cite{aumayr2022thora} that ensures any number of disjoint channels can be atomically updated within a constant time interval. 


We first note that existing HTLC-based solutions like \cite{poon2015lightning,MalavoltaMKMR19,MazumdarR22} can guarantee atomicity of the channel updates only when the subgraph induced by the set of channels form a path. 
Moreover, these solutions require users along the path to lock their payment amount for a time linear in the length of the path. 
For these reasons, in \sys we leverage Thora, which is solution that performs the multi-channel updates atomically in constant time \cite{aumayr2022thora}.
Furthermore, Thora operates in any graph structure as it does not depend on the PCN topology or the connectivity among channels.

In more detail, let $\vf$ denote the flow in which all channels need to be updated atomically. 
Let us denote the support of the flow by $E^* = \{e_i = (s_i,r_i)\}_{i=1}^m$, where $s_i$ denotes the sender and $r_i$ denotes the receiver in the $i$th channel. 
In Thora, each receiver first creates and signs a special transaction, $\mathsf{tx^{ep}}$, which contains ``dummy'' outputs for every receiver $r_i, i \in [n]$ of value $\epsilon$ and requires a signature from each sender $s_i$. 
Each sender $s_i$ then creates a transaction to update the state of channel $e_i$ and another payment transaction for their corresponding receiver $r_i$. 
A payment transaction between $s_i$ and $r_i$ takes as input the output to $r_i$ in any $\mathsf{tx^{ep}}$, and must be spent by a timeout period $T$.
Thus, as long as any $\mathsf{tx^{ep}}$ is posted by any receiver to the ledger before a given timeout period of $T$, receivers can use their corresponding output in the $\mathsf{tx^{ep}}$ to ensure they get their payment and spend it.
If any receiver does not spend their payment transaction by the timeout period $T$, all senders will get a refund with the refund amount being the amount of the payment.
In this way, atomicity of channel updates is ensured.

In Thora, the timeout period $T$ is independent of $m$, the number of channels involved.
In particular, if all users are honest, all channel updates and payments can be instantaneous, while  security is guaranteed even in the presence of malicious adversaries.
We note that Thora guarantees value privacy with respect to channel balances, but requires all users involved to know the public keys of the other involved users. These public keys can nevertheless be made pseudonymous (see \Cref{discussion:privacy} for more details).


\subsection{Algorithm description}
\label{sec:algdescr}
We consolidate the description of \sys in Algorithm \ref{alg:traggr}. 
We present \sys from the point of view of a node $u$; be it client or hub.
The input to the protocol is the list of transactions $out(\CT,u)$ originating from $u$ and the balances of channels adjacent to $u$.
The output of \sys is the result of the atomic execution of a flow $\vf^*$ that realizes a sublist $\CT^*$ of $\CT$, which is optimal in terms of total throughput $\sum_{\vt\in\CT} |t|$. 

Lines \ref{algl:secretsharing}--\ref{algl:MPCoutput} consist of the flow computation phase.
In line \ref{algl:secretsharing} node $u$ secret shares its local input to the hubs (MPC delegates).
That is, the outgoing transactions (or a single zero transaction if there is no outgoing transaction) and the capacities of all adjacent channels. 
Then, the hub nodes run lines \ref{algl:ifhub}--\ref{algl:MPCoutput}.
They first check if the incoming and outgoing transactions submitted are feasible with the given channel capacities (lines \ref{algl:inputvalidationstart}--\ref{algl:inputvalidationend}) and exclude user inputs otherwise.
Subsequently, they run the optimization oracle (line \ref{algl:oracle}), and return the output to all nodes (line \ref{algl:MPCreturn}).
With the given output (line \ref{algl:MPCoutput}), node $u$ proceeds to the execution phase.
It first validates the output in lines \ref{algl:validationstart}--\ref{algl:validationend} and then participates in Thora with its local input $f^*(u)$ (line \ref{algl:thora}).

\section{Formal analysis}
\label{sec:analysis}
In this section, we prove that \sys satisfies the protocol goals.
We first discuss the complexity of the transaction aggregation problem, and show that the presented algorithm is fixed-parameter linear in the number of hubs.
Later, we demonstrate that \sys satisfies balance security, optimality and cost efficiency.
Lastly, we prove security with our indistinguishability game as per Definition~\ref{def:privacy}.

\subsection{Complexity of the Transaction Aggregation Problem}
\label{sec:oracle}

In section~\ref{sec:computation}, we referred to an optimization oracle $\CO$ to abstract out the algorithmic challenge of computing an optimal sublist $\CT^* \subset \CT$. Here, we will first show that the problem is NP-hard in general via reduction from the well-known subset sum problem. Then we will restrict attention to our specific graph and show that the problem is fixed-parameter linear. That is, we show that the computational complexity of the optimization oracle is linear in the number of transactions $\len{\CT}$ but exponential in $m$, the number of hubs.

\begin{theorem}
The Transaction Aggregation Problem is NP-hard.
\end{theorem}

\begin{proof}

In fact, the problem is NP-hard even for a graph on two vertices and no edges. Consider an instance of the Subset Sum problem, well-known to be NP-hard: given a positive integer $A$ and a set $\set{a_1, \ldots, a_k}$ of positive integers, find a subset of them that sum to $A$.

Let $G=(V,E)$ with $V = \set{v_1, v_2}, E = \emptyset$. Let $\CT = \{(a_1, -a_1)$, $(a_2, -a_2), \ldots (a_k, -a_k), (-A, A)\}$. Finding \textit{any} nonempty subset of $\CT$, let alone one optimal in terms of throughput, which can feasibly be aggregated, yields a solution of the Subset Sum problem.
\end{proof}

\paragraph{Remark:} Transaction aggregation coincides with BCP when the graph has only two vertices. Thus, the NP hardness (and in fact inapproximability) follows from the NP-hardness and inapproximability of BCP with two participants\cite{GuntzerJL98}. 

We now proceed to calculating the computational complexity of the optimization oracle (Section \ref{sec:computation}). 
Let us first state the relevant theorem from \cite{eisenbrand} before proceeding to apply it to our setting.


\begin{theorem}[Theorem 8 in \cite{eisenbrand}]
\label{thm:disc}
An integer program of the form\footnote{The integer program in \cite{eisenbrand} is presented in the so-called standard form of $\MA \vx = \vb$, whereas we use the inequality form $\MA \vx \leq \vb$ here for simplicity of exposition. Any integer program can be reformulated from inequality form to standard form by the introduction of slack variables, which doubles the number of variables. This is inconsequential in our case as our dependence on the number of variables is linear.}
$$\max \vw \Tr \vx \quad \text{ such that}$$
$$\MA \vx \leq \vb, 0 \leq \vx \leq \vu, \vx  \in \Z^k.$$
for $\vc \in \Z^k, \MA \in \Z^{h \times k}, \vb \in \Z^h, \vu \in \Z^k_{+}$,
can be solved in time $O(k(h\Delta)^{h^2})$, provided all entries of $\MA$ are bounded in absolute value by $\Delta$.
\end{theorem}

\begin{theorem}\label{thm:txagrr}
The Transaction Aggregation problem of \sys can be solved in time $O(k(h\Delta)^{h^2})$,
where $\Delta$ is an upper bound on the demand of every transaction, $k$ is the number of transactions in $\CT$ and $h$ is the number of hubs.
\end{theorem}

\begin{proof}
It is sufficient to reformulate the transaction aggregation problem as one satisfying the hypothesis of Theorem~\ref{thm:disc} above.

Consider an arbitrary ordering of the transactions in $\CT = \{\vt_1, \ldots$, $\vt_k\}$. For a vector $x \in \set{0,1}^k$, let $\CT(\vx) = \set{\vt_i : \vx_i =1}$. We need to construct a matrix $\MA$ and vectors $\vw, \vb$ such that finding the optimal sublist $\CT^* \subset \CT$ is equivalent to finding a vector $\vx$ maximizing $\vw \Tr \vx$ subject to $\MA \vx = \vb$ and $0 \leq \vx \leq \vu$, where $\vu$ is the all ones vector.

The main idea behind the construction of $\MA$ is that, by our assumption on client-hub channels having sufficient capacity to route all transactions from $\CT$, we need only consider feasibility of the demand vector of a sublist $\CT'$ with respect to the channel factory between the hubs. 

So let us order the hubs as $h_1, h_2, \ldots, h_m$ and define $\vb = (b_1, b_2,\ldots$, $b_m)$ as: $b_i$ is the current balance of $h_i$ in the channel factory. Now, for every transaction $\vt_i \in \CT$ of demand $\len{\vt_i}$, we define $\MA_i$, the $i$th row of $\MA$ coordinate-wise as: $\MA_{i,j}  = \len{\vt_i}$ if the sender of $\vt_i$ is connected to the hub $h_j$. $\MA_{i,j} = - \len{\vt_i}$ if the recipient of $\vt_i$ is connected to hub $h_j$, and $0$ otherwise. We define $\vw$ as $\vw_i = \len{\vt_i}$.

The vector $\MA \vx$ represents the requisite flow in the channel factory for routing the aggregate of $\CT(\vx)$. Due to our assumption on the capacity of client-hub channels, this is sufficient for checking for the feasibility of the aggregate. Since $\vb$ represents the current state of the channel factory, the aggregate of $\CT(\vx)$ is feasible if and only if $\MA \vx \leq \vb$. We conclude the proof by observing that maximizing the objective $\max \vc \Tr \vx$ is equivalent to maximizing the throughput $\sum \len{\vt_i}$ of the aggregate.
\end{proof}

Theorem~\ref{thm:txagrr} immediately implies the following.
\begin{corollary}
\sys satisfies computational feasibility as in Definition~\ref{def:efficient}.
\end{corollary}

\subsection{Balance Security}
Here, we show that the users of \sys do not lose any money through their participation in the protocol.

\begin{theorem}
\sys satisfies balance security as in Definition~\ref{def:balsec}.
\end{theorem}
\begin{proof}
The flow execution phase of \sys is the only phase that involves monetary flow, hence requiring atomicity.
In \sys, flow execution is done using Thora, which guarantees atomicity of channel updates assuming rational adversaries (see \cite{aumayr2022thora} for more details). 
Thus, \sys also inherits atomicity against rational adversaries. As a result, either the atomic execution of a flow that routes our optimal list $\CT^*$ succeeds, implies that the change in user balances corresponds to the same sublist $\CT^* \subset \CT$. Or, the atomic execution failed, in which case the user balances are unchanged, corresponding to the empty sublist $\emptyset \subset \CT$.
\end{proof}

\subsection{Optimality \& Cost Efficiency}
Here, we first show that the selected set transaction is maximal in terms of throughput. Next, we show that the users of \sys may benefit from their participation in the system in terms of fees.
\begin{theorem}
\sys satisfies optimality as in Definition~\ref{def:optimality}.
\end{theorem}
\begin{proof}
This follows immediately from the objective function of the transaction aggregation  problem (Definition~\ref{def:transactionAggr}) and the correctness of the optimization oracle (Theorem~\ref{thm:txagrr}).
\end{proof}

\begin{theorem}
\sys satisfies cost efficiency as in Definition~\ref{def:costef}.
\end{theorem}
\begin{proof}
We wish to prove that total fee levied to users should be no greater than the fees that would have been levied had the transactions been processed sequentially. Although \sys routes an optimal sublist $\CT^*$, we prove the above for any list $\CT $ of transactions that can feasibly be aggregated. The main idea of this proof is that forwarding fees satisfy the triangle inequality, and that transactions can only be routed along unique paths (Section~\ref{sec:additional assumptions}).

Suppose $\CT =  [\vt_1, \ldots \vt_k]$ and the flow $\vf$ routes $\CT$. Note that a sequential execution of transactions may not be possible although the atomic execution of the aggregate (this is in fact one of the benefits of transaction aggregation), in which cost efficiency holds trivially. Suppose that a sequential execution of transactions is indeed possible, and without loss of generality assume this order is $\vt_1, \vt_2 \ldots \vt_k$. 

Consider any channel or channel factory $C$ in our PCN, and let $a_0$ be its initial state. Suppose the atomic execution of $\vf$ leaves $C$ in the final state $a_f$, and this incurs fee $F = F(a_0, a_f)$. Next, suppose the sequential execution of transactions takes $C$ along states $a_0, a_1, \ldots a_k = a_f$, so that the routing of $\vt_i$ involves a state transition from $a_{i-1}$ to $a_i$ and costs $F(a_{i-1}, a_i)$. Since the fee function satisfies the triangle inequality,
$$F(a_0, a_f) = F(a_0, a_k) \leq \sum\limits_{i=1}^k F(a_{i-1}, a_i).$$
Observing that the total fees incurred is simply the sum of forwarding fees per channel, this establishes cost efficiency.
\end{proof}

\subsection{Privacy}
Our transaction aggregation solution as well as the  execution protocol 
satisfies our privacy notion, i.e.,
\begin{theorem}
\sys satisfies the privacy notion from Definition \ref{def:privacy} (assuming the users running the MPC protocol satisfy the trust assumptions required by the underlying MPC protocol).
\end{theorem}
\begin{proof}
Consider the indistinguishability game as defined in \Cref{def:privacy}. 
Let us assume during the challenge phase of the game there exists a trusted third party (TTP) that computes a flow $\vf_b$ from the challenge transaction tuple list $\mathcal{T}^b$. 
Suppose additionally that the computed flow satisfies all the constraints as detailed in \Cref{def:privacy}. 
The TTP subsequently lets each user in $V$ know how much they need to send or receive on each of their incident edges. 

Each corrupted user learns, from the TTP, only the amount they need to send or receive on each of their incident channels.
Since the flow restricted on $G^{\mathcal{A}}$ has to be the same when the protocol is run on either $\mathcal{T}^0$ or $\mathcal{T}^1$, the joint view of the adversarial nodes when the protocol is run on $\mathcal{T}^0$ or $\mathcal{T}^1$ is identical (and thus trivially indistinguishable).

Now we can replace the above assumption of a TTP with an MPC protocol 
that computes the same functionality with the same privacy guarantees. If the parties running the MPC protocol satisfy the trust assumption of the underlying protocol, privacy is preserved. A natural choice would be to use a computationally (rather than information theoretically) secure MPC protocol (so we have privacy unless all the users are corrupted), and (for efficiency reasons) pick a small random subset of the users to run the MPC protocol after receiving the secret-shared inputs from all the participating users. This way privacy is preserved unless the entire randomly picked subset of users has been corrupted by the adversary.

The execution phase of our protocol uses the Thora protocol as a black box, and Thora guarantees that payment values along any channel that needs to be updated are not known to anyone except the channel owners.
This again ensures that involved users only know the payments and hence flow on their incident channels.
Thora leaks the set of involved users, but they can be made pseudonymous using ephemeral keys (see \Cref{discussion:privacy}). 
Thus, the only information involved users gain is the number of all other involved users, which is allowed in our privacy definition.

Since we restrict the set of involved users to be exactly the same when the protocol computes the flow given transaction tuples $\mathcal{T}^0$ and $\mathcal{T}^1$, the view of an adversarial node when the execution phase of the protocol is run on $\mathcal{T}^0$ or $\mathcal{T}^1$ is again indistinguishable, hence the execution phase is private.
\end{proof}

\section{Discussion \& Extensions}\label{sec:discussion}
In this section, we first discuss the \sys architecture and its various components. Then, we elaborate on our privacy definition and potential extensions. Finally, we revisit the assumptions of this work and in particular the fees requested for routing by the hubs.

\subsection{\sys modularity}
\sys implements the skeleton protocol described in Section~\ref{subsec:overview}. However, the specific components of the skeleton protocol are modular, in the sense that they can be replaced with other protocols that maintain the same guarantees. 

As mentioned previously, there are various suitable MPC frameworks we can use. As an active area of research, we expect improvements on the side of secret sharing and MPC protocols, which would directly apply to \sys. The delegates conducting the MPC can also be chosen from any set of parties that have sufficient computing power and strong identities. The MPC may even be entirely replaced by a semi-trusted third party if the application does not require privacy of transactions and channel balances.

The rest of the components are also replaceable. Executing the flow derived from aggregating transactions can be implemented by any protocol, like Thora, that provides atomicity of channel updates.

In \sys we assume the hubs have created an on-chain channel factory. Although this construction benefits the protocol in terms of throughput, lower liquidity requirements as well as runtime, this structure is by no means restrictive or necessary for our protocol execution. Hubs that are already connected via traditional payment channels in an arbitrary graph structure can implement transaction aggregation without paying for constructing channel factories.

Most importantly, the optimization oracle is also replaceable. In Section~\ref{sec:netting} below, we demonstrate the equivalence of BCP, the Bank Clearing Problem and transaction aggregation on channel factories. Due to the modularity of \sys, we can use any of the algorithms developed for BCP as our optimization oracles. We also discuss applicable algorithms below. 

\subsection{Connection to Netting}
\label{sec:netting}

As observed in Section~\ref{sec:WiserTopology}, the use of channel factories between $h$ hubs with capacities $\vc = (c_1, c_2, \ldots c_r)$ allows for routing  of any demand vector $\vd = (d_1, d_2, \ldots d_r)$ so long as $\vd \leq \vc$. Observe that this feasibility constraint is identical to that of the Bank Clearing Problem when $h$ participants have capital (or cover money) $\vc$. Since both transaction aggregation and BCP aim to maximize throughput or clearing volume, the problem of transaction aggregation on a channel factory is equivalent to BCP.

On one hand, this means our fixed-parameter linear algorithm from Section~\ref{sec:analysis} also shows that BCP is fixed-parameter linear in the parameter $(h\Delta)$. This linear dependence on the number of transactions is not only best possible (since simply reading a list of $k$ transactions requires $O(k)$ time), but also improves upon the previous known polynomial dependence on $k$.

Although the algorithm we present in Section~\ref{sec:oracle} enjoys the best known asymptotic complexity, other, simpler algorithms such as the pseudopolynomial time dynamic programming solution of \cite{GuntzerJL98} may be faster in certain parameter regimes. Moreover, although BCP has been shown to be inapproximable unless P = NP, \cite{ShafranskyD06} propose a fast heuristic for approximately solving it, which may also be employed here to achieve even faster implementations at the cost of optimality.

\stnote{
The computational problem of txn agg  is however more general since feasibility of the demand vector is constrained by payment channels and routing of flow on PCNs.}

\subsection{Privacy}\label{discussion:privacy}
Here we discuss two privacy issues regarding our protocol. 

\paragraph{Revealing the full set of involved users.}
As briefly mentioned in \Cref{def:privacy}, we use Thora to execute the computed flow atomically, and this leaks the full set of involved users. 
This is because every receiver in Thora has to create a special transaction $\mathsf{tx^{ep}}$ that contains outputs for every other receiver and signed with a signature from each sender. Not only is it possible, but it is also good practice, to use fresh ephemeral public keys (also known in the literature as receive addresses) when receiving tokens in UTXO-based cryptocurrencies.
If all users employ fresh receive addresses each time they participate in our protocol, the only information leaked would be the number of involved users. We consider this a non-issue, as the set of involved users is not revealed, but only the size of this set. 

\paragraph{Hiding transaction requests from hubs.}
In \sys, users secret share their transaction requests and the transaction aggregation problem is solved using MPC. Thus, the flow computation stage 
conceals transactional data.

In the execution phase, the hubs see the aggregated flow output on all their incident channels, but they cannot confidently determine sender-recipient pairs or demands since the number of transactions is not known to them. 
This gives users an additional potential privacy benefit when using our protocol as compared to when the hubs sequentially execute transaction requests. 

\subsection{Optimizing for Fees}

The algorithm presented in Section~\ref{sec:oracle} seeks a feasible sublist of transactions that maximizes throughput without considering the fees incurred in the execution of flow. As a result, we can only guarantee the weak notion of cost efficiency as given in Definition~\ref{def:cost}: that the fee is no greater than the case without transaction aggregation at all ie. had the transactions been processed sequentially.

\sys can in fact achieve greater cost efficiency by modifying the optimization oracle. In fact, only the objective function, and not the algorithm itself, must be modified to a sum of throughput and total fees incurred, with a certain trade-off factor relating the importance of these two objectives. By leaving the constraints of the integer program unchanged, this extension enjoys the same runtime as given by Theorem~\ref{thm:txagrr}.

\stnote{todo for camera-ready version. Explain in more detail that since the cost/fee is a function of resultant flow we don't make it any harder to optimize. }


\section{Related work}\label{sec:related}

Payment channel networks have received much attention recently as a promising way to increase the limited transaction throughput of blockchains. Payment channels were originally introduced by Spilman~\cite{spilman2013channels}, but the first bidirectional constructions followed later with the Bitcoin Lightning Network~\cite{poon2015lightning} and the Duplex Micropayment Channels~\cite{decker2015fast}. There is currently a flurry of payment channel proposals each optimizing a different aspect~\cite{avarikioti2019brick,mccorry2019pisa,decker2018eltoo,avarikioti2020cerberus,aumayr2020bitcoin,miller2017sprites,egger2019atomic,aumayr2020generalized,dziembowski2017perun}; see~\cite{csur21crypto} for a recent survey.

\paragraph{Routing on PCNs.}
The original proposals for payment channels faced many practical challenges. A major challenge was how to design an efficient routing algorithm for a PCN: as the current balances on the network are not known, identifying a viable short path from sender to receiver proved to be difficult. Several routing protocols for PCNs have been developed to provide efficient solutions to this problem.
Flare~\cite{prihodko2016flare} and SilentWhispers~\cite{malavolta2017silentwhispers} both employ highly connected nodes to route the payments in order to improve the scalability of the routing algorithm. SpeedyMurmurs~\cite{roos2018settling} and VOUTE~\cite{RoosBS16} improve the process by leveraging a routing approach called prefix embeddings. Flash~\cite{Flash} uses a modified max-flow algorithm to find the optimal path, while Perun~\cite{dziembowski2017perun} introduces virtual channels to avoid routing through intermediaries.
On a different front, other routing discovery algorithms focus on maintaining the anonymity and privacy during the route discovery~\cite{pietrzak2021lightpir,avarikioti2021route} or the transaction execution~\cite{tripathy2020mappcn}.

However, all of these works focus on a different problem, that of efficient (and often privacy-preserving) routing in PCNs with arbitrary graph structure. We, on the other hand, simplify the network structure and examine how the users of such a graph topology can benefit from other functionalities of transaction aggregation such as canceling out and ``transferwise’’, in order to increase throughput and decrease the intermediaries’ fees.

\paragraph{Rebalancing PCNs.}
The previously mentioned routing approaches generally ignore the issue of balance depletion; channels that transfer coins mainly in one direction will soon get depleted and must top-up the balance on-chain.
The rebalancing problem was identified and addressed in Revive~\cite{khalil2017revive}, which was the first work to propose the use of rebalancing strategies. Many interesting follow-up works spawned afterwards: Pickhardt et al.~\cite{pickhardt2020imbalance} improved the balances of a PCN as a sequence of rebalancing operations of the channel funds, while Avarikioti et al.\ presented Hide \& Seek,~\cite{avarikioti2022hide} a privacy-preserving decentralized rebalancing protocol that improved upon Revive. 
Also the Bitcoin Lightning Network comes with rebalancing plugins, e.g.,  c-lightning\footnote{https://github.com/lightningd/plugins/tree/master/rebalance} and lnd\footnote{https://github.com/bitromortac/lndmanage}. 
The Spider Network~\cite{sivaraman2020high} splits payments into smaller units and routes them over multiple paths using waterfilling, aiming to efficiently route transactions while maintaining balanced channels. 
The Merchant~\cite{DBLP:conf/dappcon/EngelshovenR21} utilises adaptive fee strategies to incentivize the balanced use of payment channels, and \cite{LiMZ20} uses estimated payment demands along channels to plan the amount of funds to inject into a channel during channel creation.
There also exist several interesting approaches to save fees in payment channel networks, e.g., virtual channels~\cite{dziembowski2017perun}, and studies on the security implications of such fee mechanisms~\cite{aft20}.

Transaction aggregation is intrinsically related to the question of rebalancing in the sense that multiple flows can execute a given aggregate of transactions, and an appropriate choice of flow inevitably results in minimal channel depletion. In other words, rebalancing can be interpreted as transaction aggregation of an empty list of transactions.
However, none of these works leverage transaction aggregation to increase throughput in PCNs and hence do not enable routing transactions even in a disconnected graph topology with the newly introduced ``transferwise’’ functionality.

\paragraph{Atomic multi-hop transaction execution.}
Traditional payment channel networks like the Bitcoin Lightning network achieve the atomic execution of multi-hop payments using Hash Timelock Contracts or HTLCs~\cite{poon2015lightning,decker2015fast}. 
however, HTLC suffer from sepcific attacks such as the wormhole attack~\cite{malavolta2019anonymous}. In a recent line of work, several proposals exist that improve upon the use of simple HTLCs to guarantee the atomic execution of transactions, typically targeting efficiency and security in an extended  model~\cite{aumayr2021blitz,tairi2021anonymous,aumayr2022thora}.
These works are complementary to ours because our proposed solution requires the use of a protocol that can execute transactions in an atomic manner in general topologies where the graph can even be disconnected. Therefore, any efficiency improvements on the atomic transaction execution in generic topologies implies immediately an efficiency improvement on \sys. 

\paragraph{Netting.}
Our work is closely related to netting systems in financial markets, specifically netting for inter-bank settlements, which involve using a central entity (typically a central bank) to settle all liabilities of involved financial institutions~\cite{GARRATT2020105270}.  
The netting problem for banks, also known as the Bank Clearing Problem, is NP-complete and inapproximable. However, heuristic-based algorithms that yield approximate solutions in practice have been developed in ~\cite{GuntzerJL98} and~\cite{ShafranskyD06}. On the side of exact algorithms, a pseudopolynomial time exact algorithm for BCP was presented in~\cite{GuntzerJL98}. This exact algorithm is based on dynamic programming, but processes transactions in the same order as input. Our algorithm, also based on dynamic programming, improves upon theirs by exploiting discrepancy techniques to reorder the transactions and hence reduce time and memory requirements to linear.

Besides the algorithmic connection, netting is also related to our work in terms of privacy and decentralization. Although inter-bank netting traditionally relies on central banks as mediators, there are multiple reasons to seek decentralized alternatives. Central banks are trusted to preserve the confidentiality of transactional data, as well as perform the netting correctly, which incurs greater liability for all involved parties. As observed in~\cite{CaoYCNEH20}, finding trusted mediators for cross-border multi-currency transfer is also challenging. To this end, there are some works which aim to implement decentralised netting systems~\cite{jasper,CaoYCNEH20}. In particular, the work of Cao et al. provides a decentralized and optimal netting solution that also guarantees the privacy of payment amounts using smart contracts and zero knowledge proofs on the blockchain~\cite{CaoYCNEH20}. 
Our work, in contrast, aims to address this problem for payments on payment channel networks, and thus all our proposed solutions in \sys seek to reduce the usage of the underlying blockchain for arbitration.

\section{Conclusion}\label{sec:conclusion}
In this work, we formulated the transaction aggregation problem and presented a protocol solving it. 
Transaction aggregation can potentially increase throughput and even realize payments that were infeasible when executed sequentially. 
We presented a skeleton protocol for a hub-based business model that abstracts the components required to (i) privately compute the flows through each channel maximizing the transaction throughput, and (ii) execute those flows atomically.
We then presented \sys, an implementation of our skeleton protocol: In \sys we employed channel factories as the intra-hub channel structure. We further proposed a fixed-parameter linear algorithm to solve the transaction aggregation problem, which we execute via an MPC to maintain privacy of transactions and channel balances. Finally, \sys utilized Thora to execute the flow atomically across the PCN.
With \sys we (a) enhanced the liquidity of the PCN, (b) effectively reducing the transaction fees to the hubs, while (c) maintaining the privacy of the transaction's and channels' data.

We regard our work as a first step toward more general protocols.
Specifically, an interesting future work is to design a computationally tractable protocol for more general topologies than the hubs/clients topology we use. Algorithmic breakthroughs in the specific topology we consider would also benefit centralized or decentralized netting protocols.
Also, we can improve the privacy guarantees by anonymizing the involved users or hiding transaction requests from hubs.
Moreover, it is orthogonal to our work to compute a fee mechanism charged by the hubs. 
Note that including fees in the optimization component is straightforward and does not affect the protocol's complexity, although its effect on throughput is not understood.




\bibliographystyle{splncs04}
\bibliography{bibfile}
\end{document}

%% file: preamble.tex
\usepackage[utf8]{inputenc}
 \usepackage{amsmath,amsfonts,amssymb}  
\usepackage{pifont}
\usepackage{dsfont}
\usepackage{authblk}
\usepackage{fullpage}
\usepackage{times}
\usepackage{hyperref}
\usepackage{microtype}
\usepackage{color}
\usepackage{cleveref}
\usepackage{comment}
\usepackage[ruled, linesnumbered,vlined]{algorithm2e}
\usepackage{algpseudocode}
\usepackage{subcaption}
\usepackage{graphicx}
\usepackage{float}
\usepackage{wrapfig}

\makeatletter
\let\cl@chapter\undefined
\makeatletter

\newcommand{\Z}{\mathbb{Z}}

\newcommand{\Tr}{^\mathrm T}

\newcommand{\vb}{\mathbf{b}}
\newcommand{\vc}{\mathbf{c}}
\newcommand{\vd}{\mathbf{d}}

\newcommand{\vf}{\mathbf{f}}

\newcommand{\vt}{\mathbf{t}}
\newcommand{\vu}{\mathbf{u}}
\newcommand{\vw}{\mathbf{w}}
\newcommand{\vx}{\mathbf{x}}

\newcommand{\CO}{\mathcal{O}}

\newcommand{\CT}{\mathcal{T}}
\newcommand{\MA}{\mathsf{A}}

\newcommand{\set}[1]{\{{#1}\}}

\newcommand{\len}[1]{\lvert{#1}\rvert}

\newcommand{\sys}{\textsc{Wiser}\xspace}

\SetKwInput{KwInput}{Local input} 
\SetKwInput{KwOutput}{Local output} 
\newcommand{\basefee}{bf}
\newcommand{\propfee}{pf}

\newcommand{\mylIf}[2]{\textbf{if}~#1~\textbf{then}~#2}

\newcommand{\NP}{\textsc{NP}}

\newcommand{\FPL}{\textsc{FPL}}